\documentclass[letterpaper, 10 pt, conference]{ieeeconf}  

\IEEEoverridecommandlockouts                              
\usepackage[utf8]{inputenc}
\usepackage[english]{babel}
\usepackage{mathtools}
\usepackage{amssymb}
\usepackage{dsfont}
\usepackage{dutchcal}
\usepackage{csquotes}
\allowdisplaybreaks 
\usepackage{booktabs}
\usepackage{blindtext}
\usepackage{bbm}
\usepackage{breqn}
\usepackage{graphicx}
\usepackage{framed}
\usepackage{tcolorbox}
\usepackage{tikz}
\let\labelindent\relax
\usepackage{enumitem}
\usepackage{balance}

\usetikzlibrary{matrix,decorations.pathreplacing}
\newtheorem{definition}{\textbf{Definition}}
\newtheorem{theorem}{\textbf{Theorem}}
\newtheorem{proposition}{\textbf{Proposition}}
\newtheorem{problem}{\textbf{Problem}}
\newtheorem{lemma}{\textbf{Lemma}}
\newtheorem{remark}{\textbf{Remark}}
 
\newcommand\oprocendsymbol{\hbox{$\square$}}
\newcommand\oprocend{\relax\ifmmode\else\unskip\hfill\fi\oprocendsymbol}

\usepackage{ifthen} 
\newboolean{showcomments}
\setboolean{showcomments}{false}   
\usepackage{todonotes}

\definecolor{bleudefrance}{rgb}{0.19, 0.55, 0.91}
\definecolor{ao(english)}{rgb}{0.0, 0.5, 0.0}

\newcommand{\enrique}[1]{  \ifthenelse{\boolean{showcomments}}
{\todo[inline,color=bleudefrance]{Enrique says: #1}}{}}

\newcommand{\har}[1]{  \ifthenelse{\boolean{showcomments}}
{\todo[inline,color=ao(english)]{Haralampos says: #1}}{}}

\newcommand{\addcite}[0]{\ifthenelse{\boolean{showcomments}}
{\textcolor{purple}{(add cite(s)) }}{}}%

\newboolean{showedits}
\setboolean{showedits}{false}
\usepackage[markup=underlined]{changes}
\definechangesauthor[color=bleudefrance]{EM}
\newcommand{\aem}[1]{
\ifthenelse{\boolean{showedits}}
{\added[id=EM]{#1}}
{\!#1\hspace{-4.5pt}}
}
\newcommand{\rem}[2]{
\ifthenelse{\boolean{showedits}}
{\replaced[id=EM]{#1}{#2}}
{\!#1\hspace{-4.5pt}}
}
\newcommand{\dem}[1]{
\ifthenelse{\boolean{showedits}}
{\deleted[id=EM]{#1}}
{}
}

\setboolean{showcomments}{false}

\newcommand{\R}{\mathbb{R}}

\renewcommand{\baselinestretch}{.9250}

\title{\LARGE \bf
Voltage Collapse Stabilization: A Game Theory Viewpoint
}

\author{Charalampos Avraam$^*$, Jesse Rines$^*$, Aurik Sarker$^*$, Fernando Paganini$^\dagger$, and Enrique Mallada$^*$ \thanks{$^*$~C. Avraam, J. Rines, A. Sarker, and E. Mallada are with The Johns Hopkins University, Baltimore, Maryland, USA. E-mail:
{\tt\footnotesize \{cavraam1,jrines1,asarker1,mallada\}@jhu.edu}.
}%
\thanks{
$^\dagger$~F. Paganini is with the Universidad ORT Uruguay, Montevideo. E-mail: {\tt\footnotesize paganini@ort.edu.uy}. }
}

\begin{document}

\setlist[itemize]{leftmargin=*}

\maketitle
\thispagestyle{empty}
\pagestyle{plain}

\begin{abstract}
Voltage collapse is a type of blackout-inducing dynamic instability that occurs when the power demand exceeds the maximum power that can be transferred through the network. The traditional (preventive) approach to avoid voltage collapse is based on ensuring that the network never reaches its maximum capacity. However, such an approach leads to inefficiencies as it prevents operators to fully utilize the network resources and does not account for unprescribed events.
To overcome this limitation, this paper seeks to initiate the study of \emph{voltage collapse stabilization}.

More precisely, for a DC network, we formulate the problem of voltage stability as a dynamic problem where each load seeks to achieve a constant power consumption by updating its conductance as the voltage changes. We show that such a system can be interpreted as a dynamic game, where each player (load) seeks to myopically maximize their utility, and where every stable power flow solution amounts to a Local Nash Equilibrium.

Using this framework, we show that voltage collapse is equivalent to the non-existence of a Local Nash Equilibrium in the game and, as a result, it is caused by the lack of cooperation between loads. Finally, we propose a \emph{Voltage Collapse Stabilizer} (VCS) controller that uses (flexible) loads that are willing to cooperate and provides a fair allocation of the curtailed demand. Our solution stabilizes voltage collapse even in the presence of non-cooperative loads. 
Numerical simulations validate several features of our controllers.
\end{abstract}

\section{Introduction}\label{sec:intro}
Voltage collapse (VC) is a type of outage in power networks that arises when the aggregate power demand exceeds the capacity of the network to transfer the required power \cite{kundur1994},\cite{vvc1998},\cite{cigre2004}. When such a point is achieved, (inflexible) constant power loads tend to rapidly reduce their effective impedance bringing the voltage abruptly to zero.
While this mechanism is intrinsically dynamic, associated with a saddle node bifurcation \cite{dobson1989},\cite{dobson1993}, the inability to correct this behavior from the generation side has lead power engineers to take a rather static (preventive) approach to address it. That is, to ensure that the point of maximum network loading is never reached \cite{iris1997}.
%
As a consequence, there has been a vast body of work trying to quantify voltage stability margins. This includes classical works, such as \cite{obadina1988}, \cite{tvc1991},\cite{dobson1994},\cite{bompard1996}  and more recently, \cite{vournas2016}, \cite{simpson-porco2016}. However, this approach leads to inefficiencies as it prevents operators to fully utilize the network resources and does not account for unprescribed events that can still produce a blackout.  

This work seeks to initiate the study of \emph{voltage collapse stabilization}. More precisely, we aim to investigate how to use (flexible) demand response to reduce consumption to match network capacity --when the total demand exceeds it-- and prevent inflexible demand from driving the system to voltage collapse. To the best of our knowledge, this work is the first effort on addressing the dynamic aspect of voltage collapse to design controllers aimed at preventing it. Such control schemes requires to overcome two main challenges. 
Firstly, it needs to stabilize an operating point that under inflexible load behavior is unstable.\footnote{At a saddle node bifurcation a stable and an unstable equilibria are merged, which leads to an unstable equilibrium \cite{khalil2002nonlinear}.} Secondly, it needs to prevent collapse even in the presence of inflexible loads.

The work is motivated by the rapid development of power electronics and information technology \cite{bayoumi2015} that, for the first time since the power system inception, has the potential to provide enough demand-side controllability that could allow to envision the possibility of stabilizing voltage collapse. However, despite the additional flexibility that controllable demand provides, there are numerous questions that remain to be answered. Among them:
\begin{itemize}
\item Is voltage collapse stabilization possible? 
\item Can stabilization be achieved via decentralized actions? 
\item {How should we allocate the necessary demand reduction among the flexible loads?}
\end{itemize}

In this paper, we build a game theoretic framework to  investigate these questions in the context of direct current (DC) networks. More precisely, we consider a star resistive DC network where each load seeks to consume a constant power by dynamically updating its conductance using a standard voltage droop. We show that such system can be interpreted as a dynamic game, where each player (load) seeks to (locally) maximize its utility, and where every stable power flow solution amounts to a Local Nash Equilibrium (LNE) (Section \ref{sec:vc-as-game}). Interestingly, voltage collapse can then be interpreted as the consequence of selfish actions of non-cooperative demand, which leads to the need of introducing coordination to overcome it.

The rest of the paper is organized as follows. Section \ref{sec:prelim} introduces our DC network model of constant power loads as well as some required game theory terminology.
Section \ref{sec:vc-as-game} frames our network model as a load satisfiability game where the unique Nash Equilibrium (NE) is the voltage collapse state, and shows that stable power flow solutions are LNE that prevent myopic players to reach their NE. Section \ref{sec:vc-stab-control} describes our voltage collapse stabilizer controller and studies its static and dynamic properties. 
We illustrate several features of our controllers using numerical simulations in Section \ref{sec:simulations} and conclude in Section \ref{sec:conclusions}.

\section{Preliminaries}\label{sec:prelim}
In this section we introduce the network model to be considered in this paper as well as the game theoretic framework to be used.
\subsection{DC Power Network Model}
\begin{figure}[htp]
  \centering
  \includegraphics[width=.9\columnwidth]{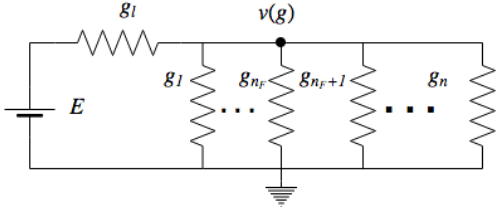}\hspace*{-1.0cm}			   
  \caption{Star DC Network with Several Dynamic loads}
  \label{fig:nbs}
 \end{figure}
We consider the star DC network model described in Figure \ref{fig:nbs},
where $E$ denotes the source voltage, and $g_l$ the conductance of a transmission line that transfer power to $n$ loads and $g_i$ denotes the $i$th load conductance, $i\in N:=\{1,\dots,n\}$. We consider two types of loads, the \textit{flexible} loads, 
belonging to the set $F=\{1,...,n_F\}$ ($n_{\text{F}}=|F|$), and the \textit{inflexible} loads, belonging to $I=\{n_F+1,...,n\}$, $n_I=|I|$. Hence, the set of all loads is $N=I\cup F=\{1,\dots, n\}$. 
We further use $g=(g_1,\dots, g_n)\in\mathbb R_{\geq0}^n$ to denote the vector of conductances
and $g_{-i}\in\mathbb R^{n-1}$ the vector of all load conductances except $g_i$. 

Using this notation, we can use Kirchoff's voltage and current laws (KVL and KCL) to compute the voltage applied to each load
\begin{equation}\label{eq:vg}
v(g)=\frac{Eg_l}{\sum_{i\in N}g_i+g_l}.
\end{equation}
Thus, the total power consumed by each load $i\in N$ becomes
\begin{equation}\label{eq:Pi}
P_i(g)=v^2(g)g_i
=\left(\frac{Eg_l}{g_{\text{eq}}(g)+g_l}\right)^2g_i,
\end{equation}
where $g_{\text{eq}}(g)=\sum_{i\in N}g_i$ is the equivalent conductance. The difference between the power consumed by each load $i\in N$ and its nominal demand $P_{0,i}$ is
\begin{equation}\label{eq:DPi}
\Delta P_i(g) = P_i(g) - P_{0,i}.
\end{equation}

The total power consumed by all the loads in the system is
\begin{equation}\label{eq:Ptot}
P_{\text{tot}} (g)=\sum_{i=1}^nP _i(g)
= \frac{(Eg_l)^2}{(g_{\text{eq}}(g)+g_l)^2}g_{\text{eq}}(g).
\end{equation}
For an arbitrary set $S \subset N$, the aggregate power consumed by every $i\in S$ is
\begin{equation}\label{eq:Ps}
P_{S} (g)=\sum_{i\in S}P _i(g)
= \frac{(Eg_l)^2}{(g_{\text{eq}}(g)+g_l)^2}g_{S},
\end{equation}
where $g_{S}=\sum_{i\in S}g_i$. Notice that if $S^c:=N \backslash S$ and $g_{S^c} =\sum_{i\in S^c}g_i$, then:
$g_{\text{eq}}(g)= g_{S}+g_{S^c}$.
\begin{definition}[Voltage Collapse]
The system \eqref{eq:infl} undergoes voltage collapse whenever $v(g(t))\rightarrow 0$  as $t\rightarrow+\infty$.
\end{definition}
\vspace{1ex}
\noindent
{ \textbf{Network Capacity} ($P_{\text{S},\max}$):}~Since voltage collapse is the result of the network reaching its maximum capacity~\cite{vvc1998}, it is of interest to compute the maximum value that $P_{\text{S}}(g)$ in (\ref{eq:Ps}) can achieve for fixed value of $g_{S^c}$.

A straightforward calculation shows that for all $i\in S$
\begin{equation}\label{eq:dPidgi}
\frac{\partial}{\partial g_i}P_i(g)=\frac{(Eg_l)^2}{(g_{{\text{eq}}}(g)+g_l)^3}\left(g_l +g_{\text{eq}}(g)-2g_i\right),
\end{equation}
which implies that
\begin{equation}\label{eq:dPsdgs}
\frac{\partial}{\partial g_S}P_{\text{S}}(g)=\frac{(Eg_l)^2}{(g_{\text{eq}}(g)+g_l)^3}\left( g_l+{g}_{S^c}-g_{S} \right).
\end{equation}

From \eqref{eq:dPsdgs}, it is easy to see that $P_S(g_S;g_{S^c})$ is an increasing function of $g_S$ up until $g_S=g_l+g_{S^c}$, and decreasing for all $g_S>g_l+g_{S^c}$.
Therefore, the maximum power that can be supplied to the loads in $S$ is given by
\begin{align}\label{eq:Ps,max}
P_{\text{S,max}}=P_{S}(g_S^*;g_{S^c})=\frac{E^2g_l}{4}\frac{g_l}{g_l+{g}_{S^c}}
\end{align}
and is achieved whenever $g_{S}^*=\sum_{i\in S}g^*_i =g_l+{g}_{S^c}$. 

In the special case where $S=N$, (\ref{eq:Ps,max}) becomes: 
\begin{align}\label{eq:Pmax}
P_\max=P_{\text{tot}}(g_S^*;g_{S^c})=\frac{E^2g_l}{4}
\end{align}
and it is achieved by  $g^*_N=g_{\text{eq}}(g^*)=\sum_{i\in N}g^*_i =g_l$. 

\vspace{1ex}
\noindent
{\bf Dynamic Load Model:}~We assume that each load $i\in N$ has a constant power demand $P_{0,i}$. For an \emph{inflexible load} $i\in I$, this demand $P_{0,i}$ must always be satisfied. This is achieved by dynamically changing the conductance $g_i$ in order to change the power consumption $P_i(g)$. Following \cite{vvc1998}, we use the following dynamic model
\begin{align}\label{eq:infl}
\dot{g}_i =&  - ( v^2(g)g_i - P_{0,i} ) = - \Delta P_i(g), & i\in I.
\end{align}

For the case of \emph{flexible loads}, we assume that although they aim to satisfy their own constant power demand $P_{0,i}$, at the same time they are willing to consume less than $P_{0,i}$ whenever $P_{0,\text{tot}}:=\sum_{i\in N}P_{0,i}>P_\max$. Thus, our goal it to design a control law

\begin{equation}\label{eq:fl}
\dot{g}_i = u_i, \quad 	 i\in F,
\end{equation}
where the input $u_i$ is such that in equilibrium 
$\Delta P_i(g)=0$ whenever $P_{0,\text{tot}}<P_{\text{max}}$.

\vspace{1ex}
\noindent
{\bf Power Flow Solutions:} Given an equilibrium $g^*$ of \eqref{eq:infl}-\eqref{eq:fl}, there exists a unique voltage $v(g^*)$ and power consumption $P(g^*)=(P_i(g^*)$, $i\in N)$. The pair $(v,P)$ is referred as \emph{power flow solution}. Thus, given the one-to-one relationship between $g$ and the pair $(v,P)$, we refer to $g^*$ as a power flow solution.

\subsection{Game Theory}

We now present the game theoretical preliminaries that will allows us to 
better grasp the level of coordination required to prevent voltage collapse. 

\begin{definition}[Normal Form Game \cite{fudtir1991game}]\label{def:noncoopgame}
A Normal Form Game is given by the triple $\langle N,S,\mathbcal{u}\rangle$ where:
\begin{enumerate}
	\item $N=\{1,...,n\}$, is the set of players.
	\item $S:=S_1\times...\times S_n$, with $S_i$ being the strategy set of player $i\in N$, is the set of strategies.
	\item $\mathbcal{u}=\{\mathbcal{u}_i,i\in N\}$, where
    $\mathbcal{u}_i : S:=S_1\times...\times S_n \rightarrow \mathbb{R}$, $\forall i\in N$, is the set of payoff functions.
\end{enumerate}
\end{definition}

Given a game $\langle N,S,\mathbcal{u} \rangle$ we seek to understand the set of strategies $s=(s_1,\dots,s_n)\in S$ for which every player has no incentive to move. Moreover, since in our context it is in general difficult to understand the best response of each player, we focus on locally optimal strategies.

\begin{definition}[Nash Equlibirum \cite{fudtir1991game}]\label{def:ne}
A strategy $s^*=(s_1,..,s_n)\in S$ is a (strict) Nash Equilibrium (NE) if and only if for each $i\in N$ \begin{align}\label{eq:ne}
	& \mathbcal{u}_i(s_i^*,s_{-i}^*) > \mathbcal{u}_i(s_i,s_{-i}^*), \qquad \forall s_i\in S_i. 
\end{align}
\end{definition}
\vspace*{4pt}

\begin{definition}[Local Nash Equlibrium \cite{ratliff2016}]\label{def:localnash}
A strategy $s^*=(s_1,..,s_n)\in S$ is a (strict) Local Nash Equilibrium (LNE) if and only if for each $i\in N$ there exists an open set $W_i\subset S_i$ such that:
\begin{align}\label{eq:lne}
	& \mathbcal{u}_i(s_i^*,s_{-i}^*) > \mathbcal{u}_i(s_i,s_{-i}^*), & \forall s_i\in W_i \backslash \{ s_i^*\}.
\end{align}{}
\end{definition}
\vspace*{4pt}

Whenever the payoff functions $\mathbcal{u}_i$ are sufficiently smooth, it is possible to verify \eqref{eq:lne} using first and second order derivatives.

\begin{lemma}[Criterion for LNE \cite{ratliff2016}]\label{lem:lnecond}
Given a \aem{game} 
$\langle N,\{S_i,i\in N\},\{\mathbcal{u}_i,i\in N\} \rangle$ with \aem{doubly} continuously differentiable payoff functions, a strategy $s^*=(s_1^*,..,s_n^*)\in S$ is a strict LNE whenever 
\begin{equation}\label{eq:lnecond}
\frac{\partial}{\partial s_i}\mathbcal{u}_i(s^*)=0\text{ and }\frac{\partial^2}{\partial s_i^2}\mathbcal{u}_i(s^*)<0,\quad \forall i\in N.
\end{equation}
\end{lemma}

\section{Game Theoretical Interpretation of Voltage Collapse}\label{sec:vc-as-game}

In this section we build a game theoretical framework that provides a deeper insight on the voltage collapse phenomenon and further suggests the necessity of coordination among resources in order to prevent voltage collapse without incurring unnecessary inefficiencies.

\subsection{A Game for Inflexible Constant Power Loads}

In our formulation, the set of players is the set of loads, both conveniently  denoted by $N$. We consider here the case of inflexible loads, that is, $I=N$ and $F=\emptyset$.
For each player $i\in N$ the strategy $s_i$ is given by its conductance $g_i\geq0$. Therefore, the strategy set $S:=\mathbb{R}^n_{\geq0}$.

Following Definition \ref{def:noncoopgame}, it remains to define the utility function of each agent $i\in N$.
The following proposition motivates a particular choice of payoff function.
\begin{proposition}\label{prop:selfutil}
{Consider the game $\langle N,S,\mathbcal{u} \rangle$, where $S=\mathbb{R}_{\geq0}^n$ and for each load $i\in N$ the utility function is given by}
\begin{equation}\label{eq:selfutil}
\begin{split}
\mathbcal{u}_i(g_i;g_{-i}) =& P_{0,i}g_i + (Eg_l)^2\mathrm{ln}\left(\frac{g_{-i}+g_l}{g_i+g_{-i}+g_l}\right) \\
&- (Eg_l)^2\left(\frac{g_{-i}+g_l}{g_i+g_{-i}+g_l} -  1 \right),
\end{split}
\end{equation}
{where $g_{-i}$ denotes $g_{N \backslash \{i\}}=\sum_{j \neq i}g_j$. 
Then, the inflexible load dynamics \eqref{eq:infl} amounts to the myopic gradient dynamics 
\begin{equation}\label{eq:gradient dynamics}
    \dot g_i = \frac{\partial}{\partial s_i}\mathbcal{u}_i(s), \quad i\in I.
\end{equation}
As a consequence, if $g^*\in\mathbb{R}^n_{\geq0}$ is a LNE of $\langle N,S,\mathbcal{u} \rangle$, then it is an equilibrium of \eqref{eq:infl}.}
\end{proposition}
\begin{proof} 
From equation \eqref{eq:infl}, it follows that if $\mathbcal{u}$ is the payoff function of the game $\langle N,S,\mathbcal{u} \rangle$, then
\begin{align*}\label{eq:pflow-integral}
\frac{\partial \mathbcal{u}_i(g_i;g_{-i})}{\partial g_i} &= -  \left(\left(\frac{Eg_l}{g_{\text{eq}}(g) +g_l}\right)^2g_i-P_{0,i}\right)
\end{align*}

Integrating above expression with respect to $g_i$ gives
\begin{align*}
 & \mathbcal{u}_i= \int_{0}^{g_i} \left(P_{0,i} - \frac{(Eg_l)^2}{(s+g_{-i}+g_l)^2}s\right)ds \\
=& \int_{0}^{g_i} \left( P_{0,i} 
- (Eg_l)^2\frac{s+g_{-i}+g_l}{(s+g_{-i}+g_l)^2} \right) ds \\
&+ \int_{0}^{g_i} (Eg_l)^2\frac{g_{-i}+g_l}{(s+g_{-i}+g_l)^2}ds \\ 
=& \left[P_{0,i}s \!-\! (Eg_l)^2 \left(\ln(s\!+\!g_{-i}\!+\!g_l) 
\!+\!\frac{g_{-i}\!+\!g_l}{s\!+\!g_{-i}\!+\!g_l} \right) \right]_0^{g_i} %
\end{align*}
We retrieve (\ref{eq:selfutil}) by substituting for the limit values.
\end{proof}

\aem{
Proposition \ref{prop:selfutil} reverse engineers the utility function such that any LNE is an equilibrium of \eqref{eq:infl}. Thus, although this suggests that some of the power flow solutions --which are represented by the equilibria of \eqref{eq:infl}-- may constitute a LNE, the following theorem unveils a rather surprising fact.
}
\begin{theorem}
[Voltage Collapse is the Unique NE]\label{thm:util-inf}
Given the induced game $\langle N,\mathbb{R}_{\geq0}^n,\mathbcal{u}\rangle $ with utility given by \eqref{eq:selfutil}, the strategy $g_i\rightarrow +\infty$ $\forall i\in N$ is the unique Nash Equilibrium. 
\end{theorem}
\begin{proof} 
\aem{We first show that each player maximizes their utility by setting $g_i\rightarrow+\infty$.} Using \eqref{eq:selfutil},
\begin{align*}
&\lim_{g_i\rightarrow+\infty} \mathbcal{u}_i(g_i;g_{-i}) = \lim_{g_i\rightarrow+\infty}  (Eg_l)^2 \ln\left( \frac{g_{-i}+g_l}{g_i +g_{-i}+g_l}\right) \\
&\quad - \lim_{g_i\rightarrow+\infty} (Eg_l)^2\left(\frac{g_{-i}+g_l}{g_i+g_{-i}+g_l}-1\right)  +\lim_{g_i\rightarrow+\infty}P_{0,i}g_i \\
&=\lim_{g_i\rightarrow+\infty} P_{0,i}g_i
\!-\! (Eg_l)^2 \left(\ln \left(\frac{g_{-i}\!+\!g_l}{g_i\!+\! g_{-i}\!+\!g_l}\right)\!-\!1 \right)\!=\! +\!\infty
\end{align*}

The previous derivation assumes that all other agents decide finite conductances. In the case where any other agent $j \neq i$ is also choosing $g_j \rightarrow \infty$, then a similar computation using \eqref{eq:selfutil} gives
\begin{equation*}
\lim_{g_j \rightarrow \infty} \mathbcal{u}_i(g_i;g_{-i}) = P_{0,i}g_i + (E g_l)^2,
\end{equation*}
which implies that 
\begin{equation*}
\lim_{g_i \rightarrow \infty}\lim_{g_j \rightarrow \infty} \mathbcal{u}_i(g_i;g_{-i}) = \lim_{g_i \rightarrow \infty}P_{0,i}g_i + (E g_l)^2 =+\infty
\end{equation*}
Therefore, choosing $g_i \rightarrow +\infty$ is a strictly dominant strategy for agent $i$, i.e. it is the best possible strategy regardless of the strategy chosen by all other agents. \end{proof}

{
Theorem \ref{thm:util-inf} unveils an unusual phenomenon. The normal (and desired) operating point of the DC network in Figure \ref{fig:nbs}, that represents a stable power flow solution of \eqref{eq:infl}, is not a Nash Equilibrium of the game. We will show next, that these stable points are in fact LNE that, because of the myopic nature of \eqref{eq:gradient dynamics}, prevent players from moving towards Voltage Collapse.
}

\subsection{Stable Equilibria are LNE}

We focus now in the  case where $P_{0,\text{tot}}<P_{\text{max}}$, which ensures the existence of power flow solutions, i.e., equilibria in \eqref{eq:infl}, with $I=N$.

\begin{definition}[S-LNE]\label{def:stable-lne}
A point $g^*$ is a Stable Local Nash Equilibrium (S-LNE) if it is a Local Nash Equilibrium of the induced game $\langle N,S,\mathbcal{u} \rangle$ \textit{and} a Stable Equilibrium of (\ref{eq:infl}).{}
\end{definition}

We start by first characterizing the region of stable equilibria. 
For this reason, we consider the set
\[
M :=\{g\in\R^n: \sum_{i\in N} g_i < g_l\}.
\]

\begin{lemma}[Characterization of Stable Region]\label{lem:char-stab-reg}
A hyperbolic equilibrium\footnote{An equilibrium is hyperbolic if its Jacobian is nonsingular.} point $g^*$ of \eqref{eq:infl} is stable if and only if $g^*\in M$.
\end{lemma}
\begin{proof}

Let $g^*$ {be} an equilibrium of \eqref{eq:infl}, i.e., $\Delta P_i(g_i^*)=0$ for all $i \in N$. The Jacobian of the system is given by
\begin{align}\label{eq:dec-jacob}
J(g^*) =	&  \frac{2v^2(g^*)}{g_{\text{eq}}(g^*)+g_l}g^*\mathds{1}_n^T-v^2(g^*)\mathbb{I}_n  \end{align}
where $\mathbb{I}_n$ is the identity matrix.

Let $K_{1}(g^*) = \frac{2v^2(g^*)}{g_{\text{eq}}(g^*)+g_l}g^*\mathds{1}_n^T$.
Then, since $K_1(g^*)$ is a rank 1 matrix, it has $n-1$ eigenvalues $\lambda_i(K_1)=0$, $i\in\{1,\dots,n-1\}$, and one non-zero eigenvalue $$\lambda_n(K_1)=\frac{2v^2(g^*)}{g_{\text{eq}}(g^*)+g_l}\mathds{1}_n^T g^*=\frac{2v^2(g^*)}{g_{\text{eq}}(g^*)+g_l} g_{\text{eq}}(g^*).$$

Therefore, since the second term in \eqref{eq:dec-jacob} is an identity matrix, the eigenvalues of $J(g^*)$ are a shifted from the eigenvalues of $K_1(g^*)$ by $-v^2(g^*)$, i.e. $\lambda_i(J(g^*))=\lambda_i(K_1(g^*))-v^2(g^*)$, which gives
\begin{align*}
\lambda_i(J)=
\begin{cases}
-v^2(g^*), &  i\in\{1,..,n-1\}; \\
\frac{2v^2(g^*)}{g_{\text{eq}}(g^*)+g_l}g_{\text{eq}}(g^*)-v^2(g^*),&i=n.
\end{cases} 
\end{align*}

We can now prove the statement of the lemma.

($\Rightarrow$) If $g^*$ is an asymptotically stable hyperbolic equilibrium, then $J(g^*)$ is Hurwitz and thus:
$\lambda_n(J)<0 \Rightarrow v^2(g^*)(\frac{2g_{\text{eq}}(g^*)}{g_{\text{eq}}(g^*)+g_l}-1)<0 \Rightarrow g_{\text{eq}}(g^*)<g_l$. \\

($\Leftarrow$) If $g^* \in M$, then: $\lambda_n(J)<0$. Since all eigenvalues of $J(g^*)$ are negative, by \cite[Theorem 3.5]{khalil2002nonlinear} $J(g^*)$ is stable. \end{proof}

We are now ready to show the main result of this section.
\begin{theorem}[Characterization of Stable Equilibria]\label{lem:char-stable-lne}
Every stable equilibria of (\ref{eq:infl}) is a LNE of $ \langle N,S,\mathbcal{u} \rangle$.
\end{theorem}

\begin{proof}
Let $g^*$ be a stable equilibrium of (\ref{eq:infl}). Then by Lemma \ref{lem:char-stab-reg}, $g^*\in M$. 
\begin{align*}
\frac{\partial^2 \mathbcal{u}_i(g_i^*)}{{\partial g_i}^2} = -\frac{\partial P_i(g_i^*)}{\partial g_i} = -v^2(g^*)\left(1-\frac{2g_i^*}{g_{\text{eq}}(g^*)+g_l}\right)
\end{align*}
Therefore, since  $2g_i^*\leq 2g_{\text{eq}}(g^*)<g_{\text{eq}}(g^*)+g_l$, it follows that $1-\frac{2g_i^*}{g_{\text{eq}}(g^*)+g_l}>0$ and $\frac{\partial P_i(g_i^*)}{\partial g_i^*}>0$. Then, by Lemma \ref{lem:lnecond}, $g^*$ is a LNE.
\end{proof}

\begin{remark}
Notice that not all LNE are stable equilibria. A LNE is a point for which the diagonal elements of the Jacobian are negative, which does not guarantee that $J(g^*)$ is Hurwitz. A counter-example is the case of two loads ($N=\{1,2\}$) whose demand $P_{0,1}$, $P_{0,2}$ is such that: $0<g_1<g_l$ and $g_l<g_2<g_l+g_1$. The vector $g=(g_1,g_2)\notin M$, but:
\begin{align*}
\frac{\partial \mathbcal{u}_i^2}{\partial g_i^2}=-\frac{(Eg_l)^2}{(g_{\text{eq}}(g)+g_l)^3}(g_l+g_{\text{eq}}(g)-2g_i)<0 & ,\forall i\in N.
\end{align*}
This point is indeed a LNE, but not a stable equilibrium of (\ref{eq:infl}) since $g_1+g_2>g_l$. 
\end{remark}

\begin{figure}[htp]
    \centering
		\includegraphics[width=0.225\textwidth]{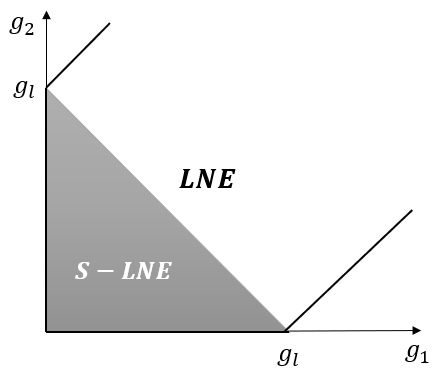}
	\caption{Region of Stable LNE for the case of two loads}
	\label{fig:regSLNE}
\end{figure}

\subsection{Voltage Collapse with Inflexible Loads}\label{ssec:voltage collapse}

We now show how in the overload regime ($P_{0,\text{tot}}>P_{\text{max}}$), the players indeed drive the system to the unique NE. 

\begin{theorem}[Voltage Collapse with Inflexible Loads]\label{prop:vc}
The dynamic load model \eqref{eq:infl} with $I=N$ undergoes a voltage collapse whenever $\varepsilon:=P_{0,\text{tot}}-P_{\text{max}}>0$.
\end{theorem}
\begin{proof}
Notice first that $\mathbb{R}^n_{\geq0}$ is invariant, since whenever $g_i=0$ \eqref{eq:infl} implies that $\dot g_i>0$. Also, it is easy to check that \eqref{eq:infl} is globally Lipschitz on $\mathbb{R}^n_{\geq0}$ since $g_l>0$. Thus by \cite[Theorem 3.2]{khalil2002nonlinear}, there is a unique solution to \eqref{eq:infl}, $g(t)$, that is defined $\forall t\geq0$. Now, consider the function $V(g) = \sum_{i\in N}g_i$, and let $S^+_V(a)=\{g\in\mathbb{R}^n_{\geq0}: V(g)\leq a\}$.
By taking the time derivative of $V$ we get
\begin{align*}
\dot V(g) &= \sum_{i=1}^n \dot g_i\! =\! -\!\sum_{i=1}^n P_i(g)-P_{0,i}
\!\geq\! P_{0,\text{tot}}\!-\!P_{\text{max}} \!=\!\varepsilon\!>\!0.
\end{align*}

Therefore, $\forall a\geq 0 $ if $g(0)\in S_V^+(a)$, $g(t)$ escapes $S_V^+(a)$ in finite time and therefore $||g(t)||\rightarrow \infty$ as $t\rightarrow \infty$. It follows then that $g_{\text{eq}}(t)$ grows unboundedly and by \eqref{eq:vg} $v(g(t))\rightarrow 0$, i.e., the system voltage collapses.
\end{proof}

\begin{remark}
Theorem \ref{prop:vc} further suggests that our model for inflexible loads successfully captures the property that excess on power demand beyond the network capacity implies voltage collapse.
\end{remark}
\subsection{Behavioral Interpretation of Voltage Collapse}\label{ssec:behavioral interpretation}

We conclude this section by illustrating how the game theoretical framework developed above allow us to obtain a behavioral interpretation of voltage collapse. 

As Theorem \ref{thm:util-inf} shows, a game representation of \eqref{eq:infl}, for which power flow solutions can provide some notion of (local) optimality (LNE), naturally leads to voltage collapse as a dominant strategy. This suggests that it is the selfish behavior of each player --that seeks to maximize their own payoff-- that constitute the underlying cause of voltage collapse.
Surprisingly, under normal operating conditions ($P_{0,\text{tot}}<P_{\text{max}}$), it is the myopic behavior of each player, who can only assess optimality within a neighborhood, that prevents the players from converging to the dominant strategy (Theorem \ref{lem:char-stable-lne}). As soon as such locally optimal solutions disappear (when $P_{0,\text{tot}}>P_{\text{max}}$), the players converge to voltage collapse (Theorem \ref{prop:vc}). 

In summary, voltage collapse is intrinsically connected to the selfish myopic desire of each load to match its required consumption, even when the network cannot provide such aggregate amount of power. This behavior is reminiscent of the one present in the \emph{tragedy of the commons} \cite{hardin1968}, and further suggests that certain level of coordination may be required in order to prevent voltage collapse. This is the basis of the solution proposed next.

\section{Voltage Collapse Stabilizer Control}\label{sec:vc-stab-control}

We now focus our attention to the task of preventing voltage collapse. Thus, we assume that there exists a subset of the loads $F\subseteq N$, $F\not=\emptyset$, that are receptive to curtailment. However, from an efficiency perspective, such curtailment should only occur whenever the total demand exceeds the network capacity ($P_{\text{0,tot}}>P_{\text{max}}$). Moreover, if curtailment does occur, it should be fairly allocated among the flexible loads.
These design objectives are summarized in the following problem formulation.

\begin{problem}[Voltage Collapse Stabilization]\label{prob:vcs}
Design a control signal $u_i$, $i\in F$, such that:
\begin{itemize}
    \item\textit{Load Satisfaction:} Whenever $P_{0,\text{tot}}<P_{\text{max}}$ the equilibrium $$g^*\in M\cap\{g:\Delta P_i(g^*)=0,\,i\in N \}$$ is the unique asymptotically stable equilibrium within $M$. 

    \item\textit{Efficient Allocation:} Whenever $P_{0,\text{tot}}>P_{\text{max}}$, the only stable equilibrium $g^*$ leads to a curtailment allocation that is the optimal solution to
    \begin{equation}\label{eq:efficiency}
    \begin{aligned}
    & \underset{\Delta P_i,i\in F}{\text{minimize}}
    & & \sum_{i\in F} \frac{\theta_i}{2}\left( \Delta P_i\right)^2 \\
    & \text{subject to}
    & & \sum_{i\in F}\Delta P_i = P_{\text{max}}-P_{0,\text{tot}}.
    \end{aligned}
    \end{equation}
\end{itemize}{}
\end{problem}

We will call a controller $u$ that solves Problem \ref{prob:vcs} a \emph{Voltage Collapse Stabilizer (VCS) Controller}. 
The rest of this section is devoted to show that the following control law is a VCS Controller: 
\begin{equation}\label{eq:vcsc}
\dot{g} = -A(g)(P(g) - P_0):\mathbb{R}^n\rightarrow\mathbb{R}^n 
\end{equation}
where $A(g)=\text{diag}\{\alpha_i(g),\,i\in N\}$,
\begin{equation}\label{eq:alpha}
\alpha_i(g_i) = 
\begin{cases}
    \frac{\kappa (\bar{g_i} - g_i)}{1+\kappa (\bar{g_i} - g_i)}, & \forall i \in F;    \\
    1, & \forall i \in I;
\end{cases}
\end{equation}
and
\begin{equation}\label{eq:bar-g}
\bar{g_i} = \frac{P_{0,i}}{(E/2)^2}+\frac{P_{\text{max}} - P_{0,\text{tot}}}{\gamma_i(E/2)^2},
\end{equation}
with $\gamma_i = \theta_i\sum_{j\in F} \frac{1}{\theta_j}$ and $\kappa$ is a positive parameter: $0<\kappa<\infty$.

\begin{remark}
The term $\alpha_i(g)$ aims to introduce a new equilibrium point $g^*$ when $P_{0,\text{tot}}>P_{\text{max}}$ such that, whenever $g^*$ satisfies $\alpha_i(g^*)=0$ $\forall i\in F$,  $\{\Delta P_i(g^*),i\in F\}$ is a solution to \eqref{eq:efficiency}. However, as we show in the next section, this can tentatively introduce new equilibria.
\end{remark}

\subsection{Characterization of Equilibria}

We now proceed to characterize the set of equilibria of \eqref{eq:vcsc}. 
Given a set of indices $G$, consider
\begin{equation}
    E_{G}:=\{g:\alpha_i(g)\!=\!0,\,i\!\in\!  G,\,\Delta P_i(g)\!=\!0,\,i\!\in\! G^c \}\label{eq:eqlbr-G}
\end{equation}
It is easy to see that the set $\{\cup_{G\subseteq F}E_G\}$ compactly encapsulates every equilibrium of \eqref{eq:vcsc}.
The following lemma will allow us to further characterize each set \eqref{eq:eqlbr-G}.
\begin{lemma}[Intermediate Value Theorem \cite{rudin1976}]\label{lem:ivt}
Let $f\in C[a,b]$. 
Then for any $\psi\in(f(a),f(b))$ there exists $\xi\in[a,b]$ such that $f(\xi)=\psi$.
\end{lemma} 

 \begin{lemma}[Characterization of $E_G$]\label{lem:char-eg}
Given any set $G \subseteq F$, the set $E_G$ composes of two equilibria 
$g^*_{1},g^*_{2}$ such that $$g_{1,i}^*=g_{2,i}^*=\bar g_i\quad\forall i\in G,$$ and 
\begin{equation}\label{eq:char-g}
    \sum_{i\in G^c}g_{1,i}^* <g_l + \sum_{i\in G} \bar g_i<\sum_{i\in G^c}g_{2,i}^*.
\end{equation}
Moreover, whenever $G=F$, then $g^*_{1}\in\text{cl}(M)$ is such that 
\begin{equation}\label{eq:char-ef}
g_{\text{eq}}(g_1^*)=g_l,\,g_{1,i}^*=\frac{P_{0,i}}{(\frac{E}{2})^2},\,\forall i\in I,    \text{ and }v(g_1^*)=\frac{E}{2}.
\end{equation}

\begin{proof}
If $g^* \in E_G$, then $\alpha_i(g^*)=0$ for all $i\in G$:
\begin{equation}\label{eq:sol-alpha-g}
    \frac{\kappa (\bar{g}_i-g_i^*)}{1+\kappa (\bar{g}_i-g_i^*)}=0\Rightarrow g_i^* = \bar{g}_i \qquad i\in G
\end{equation}
Following (\ref{eq:Ps}), let
$g_{G^c}^*=\sum_{i\in G^c}g_i^*$ and $\bar{g}_{G}=\sum_{i \in G}\bar{g}_i$. We will solve 
\begin{equation}\label{eq:eq-cond-g}
f(g_{G^c}^*)=\frac{(Eg_l)^2}{(g^*_{G^c}+\bar{g}_{G}+g_l)^2}g_{G^c}^*-P_{0,G^c} = 0
\end{equation}
for $g^*_{G^c}$. Following (\ref{eq:Ps,max}), there exists $g^*_{G^c}$ s.t. $f(g^*_{G^c})=0$ only if $P_{0,G^c}(g^*) \leq \frac{E^2g_l}{4}\frac{g_l}{\bar{g}_{G}+g_l}$, namely when the demand does not exceed the capabilities of the line. Moreover, the first and second order derivative of $f(g^*_{G^c})$ and $P_{G^c}(g^*)$ are identical, since $P_{0,{G^c}}$ is a constant. Therefore, from (\ref{eq:Ps,max}), $f(g^*_{G^c})$ is maximized for $g^*_{G^c,\text{max}}=\bar{g}_{G}+g_l$, where 
\begin{equation}\label{eq:fmax}
f_{\text{max}}=f(g^*_{G^c,\text{max}})=\frac{E^2g_l}{4}\frac{g_l}{\bar{g}_{G}+g_l}-P_{0,G^c}.
\end{equation}
In addition, $f\in C^{\infty}[0,\infty)$,
\begin{equation*}
    f(0) = - P_{0,G^c}<0, \text{ and } \lim_{g^*_{G^c}\rightarrow \infty} f(g^*_{G^c})\rightarrow - P_{0,G^c} < 0
\end{equation*}
Hence, if $f_{\max}>0$ (which is true solely when $P_{0,G^c}(g^*) < \frac{E^2g_l}{4}\frac{g_l}{\bar{g}_{G}+g_l}$), then the Intermediate Value Theorem in the two intervals $[0,g^*_{G^c,\text{max}}]$, $[g^*_{G^c,\text{max}},\infty)$ concludes that there will exist $g^*_{G^c_{1}}\in (0,g^*_{G^c,\text{max}})$ and $g^*_{G^c_{2}}\in (g^*_{G^c,\text{max}},\infty)$ that are roots of $f(g^*_{G^c})$. Thus, there exists $g_1^*$ and $g_2^*$ satisfying \eqref{eq:char-g}.

When $G=F$ (hence $G^c=I$), we have
\begin{equation*}
\begin{split}
    g_{\text{eq}}(g^*) \overset{(\ref{eq:sol-alpha-g})}{=}&\sum_{i\in F}\bar{g}_i^*+\sum_{i\in I}g_i^* \overset{(\ref{eq:bar-g})}{\!=\!} \sum_{i\in F}\frac{P_{0,i}}{(\frac{E}{2})^2} \\
   &\!+\! \sum_{i\in F}\frac{1}{\gamma_i}\frac{P_{\text{max}}\!-\!P_{0,\text{tot}}}{(\frac{E}{2})^2} \!+\! \sum_{i\in I}\frac{P_{0,i}}{(\frac{E}{2})^2} 
    \!=\! \frac{P_{\text{max}}}{(\frac{E}{2})^2}\overset{(\ref{eq:Pmax})}{\!=\!} g_l
    \end{split}
\end{equation*}
where in the unmarked equality we used fact that $\sum_{i\in F}\frac{1}{\gamma_i}=1$ by definition. 

Now, since $g_{\text{eq}}(g^*)=g_l$, it follows immediately from (\ref{eq:vg}) that $v(g^*)=\frac{E}{2}$. 
Moreover, since  $g^*$ is indeed an equilibrium we must have $P_i(g^*)=v^2(g^*)g_i^*=\left(\frac{E}{2}\right)^2g_i^*=P_{0,i},\, \forall i\in I,$ which leads to $g_i^*=\frac{P_{0,i}}{(\frac{E}{2})^2}$.
Finally, we can check that $g^*=g_1^*$ since $g^*_{I}<g_{\text{eq}}(g^*)=g_l<\bar{g}_{F}+g_l$. \end{proof}
\end{lemma}
 
 We now show that our controller \eqref{eq:vcsc} does in fact guarantee the existence of an equilibrium that solves \eqref{eq:efficiency}.
 \begin{theorem}[Efficient Allocation]\label{th:efficient curtailment}
Consider the system \eqref{eq:vcsc} with equilibria characterized by the set $E_F$ as shown in \eqref{eq:eqlbr-G}. Then, the conductance $g^*\in E_F\cap \text{cl}(M)=\{g^*\}$ leads to an efficient curtailment $\{\Delta P_i(g^*),i\in F\}$ that is  optimal w.r.t. \eqref{eq:efficiency}.
 \end{theorem}
\begin{proof}
We have shown in Lemma~\ref{lem:char-eg} that there exists $g_1^*\in E_F$ such that $g_{\text{eq}}(g_1^*)=g_l$, i.e. $g_1^*\in E_F\cap \text{cl}(M)$. For this equilibrium, the total power is 
\begin{equation*}
        P_{\text{tot}}(g_1^*)=  v^2(g^*)g_{\text{eq}}(g_1^*)\overset{(\ref{eq:char-ef})}{=}\left(\frac{E}{2}\right)^2g_l \overset{(\ref{eq:Pmax})}{=}P_{\text{max}}
\end{equation*}

For $g^*_1$ we can compute the allocation of the curtailment among loads $i\in F$:
\begin{equation}
    \begin{split}
\Delta P_i(g^*)=& v^2(g^*)\bar{g}_i - P_{0,i} \\
=& \left(\frac{E}{2}\right)^2 \left(\frac{P_{0,i}}{(\frac{E}{2})^2} + \frac{P_{\text{max}} - P_{0,\text{tot}}}{\gamma_i (\frac{E}{2})^2}\right) - P_{0,i} \\
= &\frac{P_{\text{max}} - P_{0,\text{tot}}}{\gamma_i}, \qquad \forall i \in F \\
    \end{split}
\end{equation}

We can easily check that the allocation of the curtailment is proportional to the $\theta$s
\begin{equation}
\begin{split}
\frac{\Delta P_i(g^*)}{\Delta P_j(g^*)}= \frac{\gamma_j}{\gamma_i}=\frac{\theta_j}{\theta_i}, \qquad i,j\in F
    \end{split}
\end{equation}
and thus is an Efficient Allocation.
\end{proof}

\begin{remark}
Theorem \ref{th:efficient curtailment} only guarantees that one of the equilibria of $E_F$ solves \eqref{eq:efficiency}. However, it does not provide any information regarding all the possible additional equilibria $E_G$. We will show that the remaining equilibria either do not exist, are unstable, or do not belong to the monotonicity region $M$.
\end{remark}

We conclude this section showing a case where  $E_G=\emptyset$.

\begin{theorem}[Infeasiblity of $E_G$ under extreme loading]\label{thm:infes-eg}
When $G \subset F \subseteq N$ and $P_{0,\text{tot}}>P_{\text{max}}$, then $E_G=\emptyset$.
\end{theorem}
\begin{proof}
Since $G\subset F$, then there exists a non empty set $I_F \subset F$ such that $I \cup I_F = G^c$ (or $G=F \backslash I_F$). From Lemma \ref{lem:char-eg}: $g^*_G=\sum_{i\in G}\bar{g}_i=\bar{g}_G$. Let $\bar{g}_{G^c}=g_l-\bar{g}_G$. Since $g^*\in M$, it holds that:
$$ g^*_{G^c}+\bar{g}_G<g_l = \bar{g}_G+\bar{g}_{G^c} \Rightarrow g^*_{G^c} < \bar{g}_{G^c}
$$
Notice that since $\bar{g}_{G^c}+\bar{g}_G=g_l$, (\ref{eq:vg}) implies that $v(\bar{g}_{G^c};\bar{g}_G)=\frac{E}{2}$.We will now look at how $P_{G^c}(g^*_{G^c};\bar{g}_{G})$ behaves with respect to $g^*_{G^c}$. Similarly to (\ref{eq:dPsdgs}), $P_{G^c}(g^*_{G^c};\bar{g}_{G})$ is a strictly increasing function for $g^*_{G^c}<g_l+\bar{g}_G$. Therefore:
\begin{equation*}
    \begin{split}
       P_{G^c}(g^*_{G^c};\bar{g}_{G}) -  P_{0,G^c}&    \overset{g^*_{G^c} <\bar{g}_{G^c}}{<} \left( \frac{E}{2} \right)^2 \bar{g}_{G^c} - P_{0,G^c} \\
        &\overset{(\ref{eq:Pmax})}{=} P_{\text{max}} - \left( \frac{E}{2} \right)^2\bar{g}_G - P_{0,G^c} \\
                &\overset{(\ref{eq:bar-g})}{=} (P_{max} \!-\! P_{0,\text{tot}}) \left(1\!-\!\sum_{i\in G} \frac{1}{\gamma_i}\right)\!<\! 0
    \end{split}
\end{equation*}
where in the first equality we have substituted for $\bar{g}_{G^c}=g_l - \bar{g}_G$  and in the last step we have substituted (\ref{eq:bar-g}) and rearranged the terms. However, since $g^*\in G$, then (\ref{eq:eqlbr-G}) and thus $P_{G^c}(g^*_{G^c};\bar{g}_G)=P_{0,G^c}$. The above contradiction implies that $E_G\cap M=\emptyset$.

Let $g^*\in E_G\cap M^c$, then $\sum_{i\in N}g_i^* \geq g_l$ and from (\ref{eq:vg}): $v(g^*) \leq \frac{E}{2}$. Therefore:
\begin{align*}
    & 0 \leq \sum_{i\in G} P_i(g^*) = \sum_{i\in G} v^2(g^*)g_i^* = \sum_{i\in G} v^2(g^*)\bar{g}_i  \\
    &\overset{(\ref{eq:alpha})}{=} \sum_{i\in G} P_{0,i}\left( \left( \frac{v(g^*)}{\frac{E}{2}}\right)^2-1 \right) + \sum_{i\in G} \frac{v^2(g^*)(P_{\text{max}} - P_{0,\text{tot}})}{\gamma_i \left(\frac{E}{2}\right)^2}  \\
    & < 0 
\end{align*}
Hence, if $P_{0,\text{tot}}>P_{\text{max}}$ then $E_G\cap M^c = \emptyset$. Therefore: $E_G = \left(E_G\cap M \right)\cup\left(E_G\cap M^c \right) = \emptyset $. 
\end{proof}

So far we have shown that whenever $P_{0,\text{tot}}>P_{\text{max}}$, the only feasible set $E_G$ is $E_F$, and $E_F$ contains the equilibrium that solves the efficient curtailment problem \eqref{eq:efficiency}. 
The next section will show that this is in fact the only stable equilibrium under extreme loading conditions.

\subsection{Stability Analysis}
In this Section we will study the stability of the different equilibria with the objective of showing that  the chosen controller solves Problem \ref{prob:vcs} and thus qualifies as a VCS Controller. 

The following lemma will be of use in the eigenvalue computation.
\begin{lemma}[Matrix Determinant Lemma \cite{bartlet1951}]\label{lem:mat-det-lem}
If D is an invertible $n\times n$ matrix and $v,w\in \mathbb{R}^n$, then: 
\begin{equation}
\text{det}\left( D + vw^T \right) = (1+w^T D^{-1}v)\text{det}(D)\label{eq:mat-det}
\end{equation}{}
\end{lemma} 

We can now compute the eigenvalues of the Jacobian of (\ref{eq:vcsc})
\begin{lemma}[Computation of Eigenvalues of (\ref{eq:vcsc})]\label{lem:vcsc-eigvals}
Consider the system (\ref{eq:vcsc}). Then, the eigenvalues of its Jacobian $J_C(g)$ at each point $g$ satisfy: 
\begin{equation}\label{eq:eig-jacob-vcsc}
\begin{cases}
    \lambda_i = \Delta P_i(g)\kappa, &\qquad i \in G ,\\
    \lambda_i : c(g,\lambda_i)=0, &\qquad \text{o.w..}
\end{cases}
\end{equation}
where 
\begin{equation}\label{eq:c}
    c(g,\lambda) := \left(1+\frac{2v^{^2}(g)}{g_{\text{eq}}(g)+g_l}  \sum_{i\in N}\frac{a_i(g_i)g_i}{ d_i(g) - \lambda} \right)
\end{equation}
\end{lemma}
\begin{proof}
The Jacobian of this system is
\begin{equation}\label{eq:jacob-vcsc}
J_C(g)= A(g)
\Big(
\frac{2v^2(g)}{g_{\text{eq}}(g)+g_l}g\mathds{1}_n^T - v^2(g)\mathbb{I}_n 
\Big) - D_{\Delta P}(g) D_{\kappa}(g)
\end{equation}
where
\begin{equation*}
D_{\Delta P}(g)=\text{diag}\{ \Delta P(g)\}, \text{ } D_{\kappa}(g)=\text{diag}\left\{  \frac{\partial}{\partial g_i}\alpha_i(g_i) \right\},
\end{equation*}
with $\frac{\partial}{\partial g_i}\alpha_i(g_i)= \begin{cases}
    \frac{-\kappa}{(1+\kappa(\bar{g_i} - g_i))^2} &\qquad \forall i\in F\\
    0   &\qquad \forall i\in I
\end{cases}$.

The eigenvalues of $J_C(g)$ are given as the solution of $\det(J_C(g) - \lambda \mathbb{I}_n)=0$. Notice that $J_C(g) - \lambda \mathbb{I}_n$ is composed by a diagonal matrix
$$D(g,\lambda):=-D_{\Delta P}(g)\Delta_{\kappa}(g) - v^2(g)A(g)- \lambda \mathbb{I}_n$$
plus a rank 1 matrix $-\frac{2v^2}{g_{\text{eq}}+g_l}(A(g)g)\mathds{1}_n^T=vw^T$, with $$v=-\frac{2v^2}{g_{\text{eq}}+g_l}(A(g)g),\quad w=\mathds{1}_n.$$
Moreover, the entries of $D(g,\lambda)$ can be written as $d_i(g) - \lambda$, with $d_i(g) = - a_i(g_i)v^2(g)+\Delta P_i(g)\frac{\kappa}{(1+\kappa (\bar{g_i} - g_i))^2}$. 

Therefore, using \cite{bartlet1951} we can compute  
\begin{align*}
 \det(J_C(g) - \lambda\mathbb{I}_n) = c(g,\lambda) \det ( D(g,\lambda) ),&
%
%
\end{align*}
which implies that either $\lambda_i$ is either equal to $\Delta P_i(g)\kappa$ or is a solution to $c(\lambda)=0$. Result follows. \end{proof}

Having characterized the eigenvalues of $J_C(g^*)$, we now analyze the stability of the equilibria of (\ref{eq:eqlbr-G}). 

\begin{theorem}[Stability of VCS Controller]\label{thm:stab} 
Consider the system \eqref{eq:vcsc}. Then, for $0<\kappa < \infty$ the following  holds:
\begin{itemize}
    \item[(1)\!\!]  When $\sum_{i \in N} P_{0,i} > P_{\text{max}}$, 
the only asymptotically stable equilibrium is given by $g^*\in E_F\cap \text{cl}(M)=\{g^*\}$.
\item[(2)\!\!]  When $\sum_{i \in N} P_{0,i} < P_{\text{max}}$, then the only asymptotically stable equilibrium within the closure of $M$  is given by $g^*\in E_{\emptyset}\cap \text{cl}(M)=\{g^*\}$.
\end{itemize}
\end{theorem}

\begin{proof}
%
(1) Let $\sum_{i \in N} P_{0,i} > P_{\text{max}}$. Then, by Proposition \ref{prop:vc} there does not exist $g^*$ such that $\Delta P_i (g^*)=0$, $\forall i\in N$. That is, if $G=\emptyset$ then $E_{\emptyset}=\emptyset$. Moreover, by Theorem \ref{thm:infes-eg}, if $\emptyset\not= G\subset F$, then again $E_G = \emptyset$. Therefore, when $P_{0,\text{tot}}>P_{\text{max}}$, $E_G$ is nonempty only for $G=F$. 

Now for $g^*\in E_F$, Lemma \ref{lem:char-eg} shows that $E_F$ comprises of two equilibria $g^*_1,g^*_2$. We will first examine $g_1^*$. Again, by Lemma \ref{lem:char-eg}, $g_{\text{eq}}(g_1^*)=g_l$ and by Lemma \ref{lem:vcsc-eigvals} the eigenvalues of  $\lambda_i=\Delta P_i(g^*)\kappa$ for all $i\in F$. Therefore, since
\begin{equation*}
\begin{split}
P_i(g_1^*)&=v^2(g_1^*)g_{1,i}^* = \left(\frac{E}{2} \right)^2 \left(\frac{P_{0,i}}{(\frac{E}{2})^2} \!+\! \frac{1}{\gamma_i}\frac{P_{\text{max}} \!-\! P_{0,\text{tot}}}{(\frac{E}{2})^2} \right)  \\
    &=  P_{0,i} + \frac{1}{\gamma_i}(P_{\text{max}} - P_{0,\text{tot}}),
\end{split}
\end{equation*}
it follows that $\lambda_i=\Delta P_i(g_1^*)\kappa =\frac{1}{\gamma_i}(P_{\text{max}} - P_{0,\text{tot}})\kappa<0$, for all $i\in F$. 

The rest of the eigenvalues are computed from (\ref{eq:c}) by substituting $\alpha_i(g_{1,i}^*)=0$ for $i\in F$ and $\alpha_i(g^*_{1,i})=1$ for $i\in I$:
\begin{equation}\label{eq:sec-eq}
    c(g_1^*,\lambda) = 1 + \frac{2v^2(g_1^*)}{g_{\text{eq}}(g_1^*)+g_l}\sum_{i \in I}\frac{g_{1,i}^*}{-v^2(g_1^*)- \lambda} = 0
\end{equation}
It is easy to show, following the analysis of \cite{jakovcevic2015}, that $c(g,\lambda)$ always has real roots.

We will examine the sign of the roots of $c(g^*,\lambda)$ by first looking at its derivative:
\begin{equation*}
    \begin{split}
    \frac{\partial}{\partial \lambda} c(g_1^*,\lambda) = \frac{2v^2(g_1^*)}{g_{\text{eq}}(g_1^*)+g_l}\sum_{i\in I}\frac{g_i^*}{(-v^2(g_1^*)-\lambda)^2} > 0 \\
\end{split}
\end{equation*}
Hence, since the denominator is non-singular for $\lambda\geq0$, $c(g_1^*,\lambda)$ is continuous and strictly increasing for $\lambda \in [0,+\infty)$. Moreover, $c(g_1^*,0) = 1 - \frac{2 \sum_{i \in I}g_{1,i}}{g_{\text{eq}}(g_1^*)+g_l} >0$. Therefore, there does not exist $\lambda \geq 0$ such that $c(g_1^*,\lambda) =0$. Consequently, $\lambda_i<0$ for all $i \in I$ and we have shown above that $\lambda_i<0$ also for all $i\in F$. Therefore, by Lyapunov's Inidrect Method~\cite[Theorem 3.2]{khalil2002nonlinear} $g_1^*\in E_F$ is a stable equilibrium of (\ref{eq:vcsc}). 

The analysis for $g_2^* \in E_F$ is analogous. 
The function $c(g_2^*,\lambda)$ is again continuous and strictly increasing for $\lambda\in[0,+\infty)$.
However, since by Lemma \ref{lem:char-eg}  $g_{I,2}^*>g_l+\bar{g}_{G}$, 
\begin{equation*}
c(g_2^*,0)=1-\frac{2g_{I,2}^*}{g_{\text{eq}}(g_2^*)+g_l}<0, \,\text{while }\; \lim_{\lambda\rightarrow \infty} c(g_2^*,\lambda) = 1.
\end{equation*}
Thus, by Lemma \ref{lem:ivt} there  exists $\bar{\lambda}\in(0,+\infty)$ such that $c(g_2^*,\bar{\lambda})=0$. Since $J_C(g_2^*)$ has at least one positive eigenvalue, using again \cite[Theorem 3.2]{khalil2002nonlinear} we can conclude that $g_2^*$ is unstable.

(2) Let $\sum_{i \in N} P_{0,i} < P_{\text{max}}$ and $g^*\in E_{\emptyset}\cap \text{cl}(M)$. In this case $g_i^*\neq \bar{g}_i$ $\forall i\in N$ (otherwise $g^*\notin E_{\emptyset}$). Therefore $\alpha_i(g_i^*)\neq 0$. From (\ref{eq:eig-jacob-vcsc}) the eigenvalues of the system satisfy: \begin{equation}\label{eq:sec-ef1}
    c(g^*,\lambda)= 1 +\frac{2v^2(g^*)}{g_{\text{eq}}(g^*)+g_l}\sum_{i\in N}\frac{\alpha_i(g_i^*)g_i^*}{d_i(g^*)-\lambda}=0    
\end{equation}
Given that $g^*\in E_{\emptyset}\cap \text{cl}(M)$, then $v(g^*) \geq \frac{E}{2}$ from \eqref{eq:vg}. Therefore:
\begin{equation*}
    g_i^* = \frac{P_{0,i}}{v^2(g^*)} \leq \frac{P_{0,i}}{(\frac{E}{2})^2} \overset{g^*_i \neq \bar{g}_i}{<}  \bar{g}_i\Rightarrow \alpha_i(g^*)>0 
\end{equation*}
From \cite{jakovcevic2015}, when $\alpha_i(g^*)>0$, equation (\ref{eq:sec-ef1}) has $n-1$ roots that satisfy $\lambda_i<\max_i \{ -\alpha_i(g_i^*)v^2(g^*) \}=d_M<0$. For the $n$th eigenvalue, we observe that $c(g^*,\lambda)\in C^{\infty}(d_M,0]$ and: 
\begin{align*}
    & c(g^*,d_M^-)=\lim_{\lambda\rightarrow d_M^-} c(g^*,\lambda) \rightarrow -\infty < 0  \\
    & c(g^*,0)=1-\sum_{i\in N} \frac{2g_i^*}{g_{\text{eq}}(g^*)+g_l}\overset{g\in M}{>} 0
\end{align*}
From Lemma \ref{lem:ivt}, there exists $\lambda_n \in (d_M,0)$ s.t. $c(\lambda_n)=0$. Therefore, the $n$th eigenvalue is also negative and from \cite[Theorem 3.2]{khalil2002nonlinear}, $g^*\in E_{\emptyset}\cap \text{cl}(M)$ is stable. \\
If $g^* \in E_G\cap \text{cl}(M)$ for an arbitrary $G \neq \emptyset$, then the only equilibrium that satisfies this condition is $g_1^* \in E_G$. From (\ref{eq:vg}) $v(g^*) \geq \frac{E}{2}$. Substituting into (\ref{eq:eig-jacob-vcsc}) for $i\in G$:
\begin{equation*}
\begin{split}
\lambda_i=& \Delta P_i(g^*)\kappa=(v^2(g^*)\bar{g}_i - P_{0,i})   \\
=& P_{0,i}\left( \frac{v(g^*)}{(\frac{E}{2})} - 1 \right)^2 + \frac{1}{\gamma_i}\frac{P_{\text{max}} - P_{0,\text{tot}}}{(\frac{E}{2})^2} >0 \end{split}
\end{equation*}
Since there exists at least one positive eigenvalue, \cite[Theorem 3.2]{khalil2002nonlinear} implies that the equilibrium is unstable.   
Therefore, the only stable eigenvalue is $g^*_1\in E_\emptyset\cap \text{cl}(M)$.
\end{proof}
Theorem \ref{thm:stab} shows that \eqref{eq:vcsc} is indeed a Voltage Collapse Stabilizer Controller.

\section{Numerical Illustrations}\label{sec:simulations}

In this section, we validate our theoretical results using numerical illustrations. We consider a DC grid as in Figure \ref{fig:nbs} with three loads. In all the experiments we start the simulations with initial set-points such that  $P_{0,\text{tot}}<P_{\text{max}}$ and with conductances close to the equilibrium $g^*$ where all demands are met, i.e., $g^*\in E_\emptyset\cap M$. Finally, we have chosen $\kappa=10$.

\noindent
\textbf{Case 1 (I=N=\{1,2,3\}):}
Figure \ref{fig:1vc} illustrates the behavior of the system \eqref{eq:infl}-\eqref{eq:fl} consisting of only inflexible loads. We can see that as soon as the aggregate demand reaches $P_{\text{max}}$, the system undergoes a voltage collapse.
\begin{figure}[h!]
  \centering
  	\includegraphics[width=0.4750\textwidth]{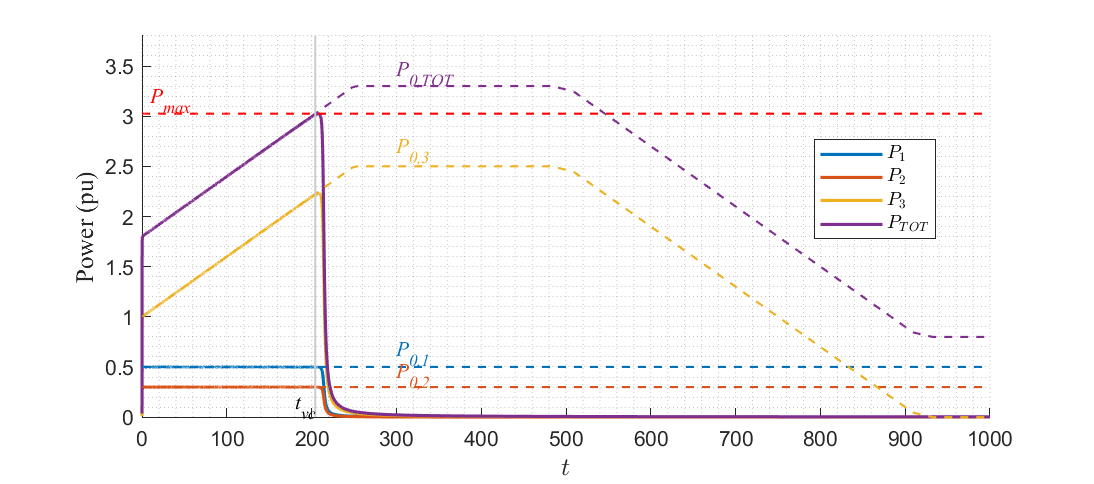}
	\caption{Voltage collapse illustration with inflexible loads.}
    \label{fig:1vc}
\end{figure}

\noindent
\textbf{Case 2 (F=N=\{1,2,3\}):} The case where all loads are flexible is illustrated in Figure \ref{fig:centr}.  In comparison with Case 1, here our VCS controller forces the consumption of \textit{all} loads to adjust proportionally to their assigned weight $\theta_i$, in this way preventing voltage collapse.
\begin{figure}[h!]
  \centering
  	\includegraphics[width=0.475\textwidth]{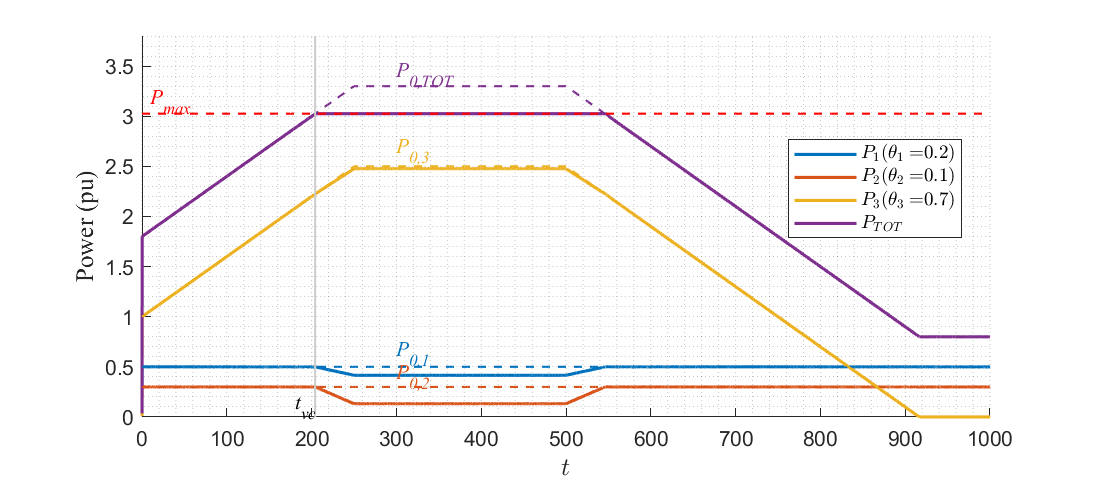}
	\caption{VCS Controller acting on all loads. Supply tracks demand under varying loading.}
    \label{fig:centr}
\end{figure}
 
\noindent
\textbf{Case 3 (F=\{1,2\},I=\{3\}):} 
Finally we illustrate a case with mixed load types where load 3 is inflexible, and our VCS controller is executed in loads 1 and 2. We observe in Figure \ref{fig:semi-centr} that the flexible loads ($1,2$) adjust their demand proportionally to their assigned weights in order to accommodate the increasing demand of the inflexible load, again preventing voltage collapse.  
However, when the system runs out of flexible demand the system will eventually undergo a voltage collapse, as predicted in Lemma \ref{lem:char-eg}.  We observe exactly this behavior in Figure \ref{fig:semi-centr-vc}.

\begin{figure}[ht!]
  \centering
  	\includegraphics[width=0.475\textwidth]{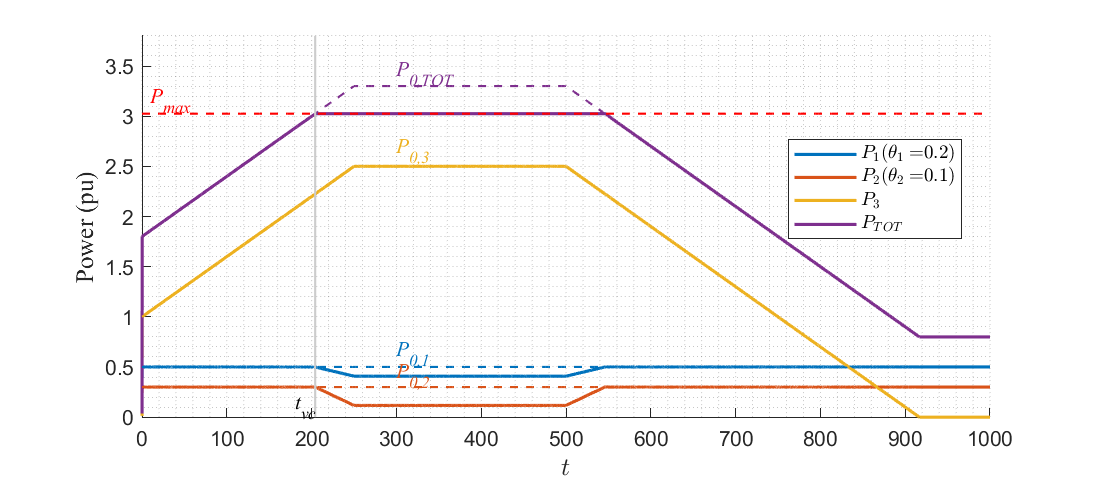}
\caption{Loads 1,2 are flexible and Load 3 is inflexible. Supply adjusts according to the assigned weight $\theta_i$}    \label{fig:semi-centr}
\end{figure}

\begin{figure}[ht!]
  \centering
  	\includegraphics[width=0.475\textwidth]{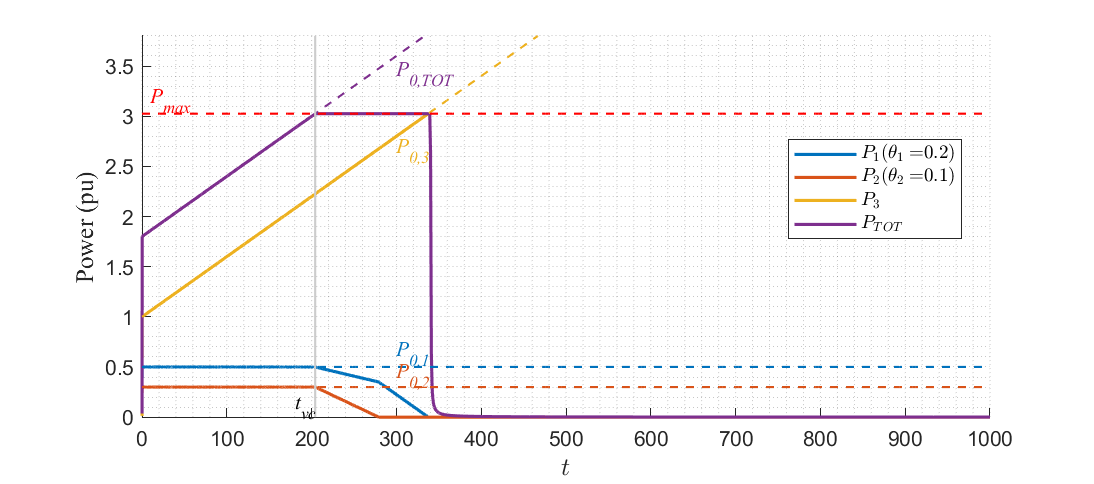}
	\caption{Loads 1,2 are flexible and Load 3 is inflexible. VC happens when $\sum_{i\in I} P_{0,i} > P_{\text{max}}$}
    \label{fig:semi-centr-vc}
\end{figure}

\section{CONCLUSIONS}\label{sec:conclusions}
This work seeks to initiate the study of voltage collapse stabilization as a mechanism to provide a more efficient and reliable operation of electric power grids. 
We develop a game theoretical framework that sheds light on the behavioral mechanism that leads to voltage collapse and suggests the need of cooperation as a means to prevent it. Based on this insight, we propose a Voltage Collpase Stabilizer controller that is able to not only prevent voltage collapse, but also fairly distribute the curtailment among the flexible loads.
Further research needs to be conducted to fully characterize the behavior of our solution. In particular, the point where $P_{0,\text{tot}}=P_{\text{max}}$ is a non-trivial point, where the Jacobian of the system is identically zero and thus requires the treatment of higher order dynamics. We identify two desired extensions of this work that are subject of current research: (a) extending the analysis to a general DC network and (b) extending the analysis to a general AC network.

\bibliographystyle{IEEEtran}
\bibliography{acc2019.bib}

\balance

\end{document}